\documentclass{CSML}
\pdfoutput=1

\usepackage{lastpage}
\usepackage{hyperref}
\usepackage{qbf-ps}

\newcommand{\dOi}{13(2:7)2017}
\lmcsheading{}{1--\pageref{LastPage}}{}{}{Nov.~07, 2016}{Jun.~08, 2017}{}

\hypersetup{
  hidelinks
}

\newtheorem{definition}[thm]{Definition} 
\newtheorem{theorem}[thm]{Theorem} 
\newtheorem{remark}[thm]{Remark} 

\begin{document}

\title[Feasible Interpolation for QBF Resolution Calculi]{Feasible Interpolation for QBF Resolution Calculi\rsuper*}
\author[O.~Beyersdorff]{Olaf Beyersdorff\rsuper a}
\address{{\lsuper{a}}School of Computing, University of Leeds, United Kingdom}
\address{{\lsuper{b}}The Institute of Mathematical Sciences (HBNI), Chennai, India}
\author[L.~Chew]{Leroy Chew\rsuper a}
\author[M.~Mahajan]{Meena Mahajan\rsuper b}
\author[A.~Shukla]{Anil Shukla\rsuper b}
\address{\vspace{-18pt}}

\begin{abstract}
  In sharp contrast to classical proof complexity we are currently short of lower bound techniques for QBF proof systems. In this paper we establish the feasible interpolation technique for all resolution-based QBF systems, whether modelling CDCL or expansion-based solving. This both provides the first general lower bound method for QBF proof systems as well as largely extends the scope of classical feasible interpolation. We apply our technique to obtain new exponential lower bounds to all resolution-based QBF systems for a new class of QBF formulas based on the clique problem. Finally, we show how feasible interpolation relates to the recently established lower bound method based on strategy extraction by Beyersdorff et al.\ \cite{BCJ15,BBC16}.
\end{abstract}

\ACMCCS{[{\bf Theory of computation}]: Proof complexity}



\keywords{Feasible interpolation, Proof complexity, QBF, Resolution}

\titlecomment{{\lsuper*}This work was  supported by the EU Marie Curie IRSES grant CORCON, 
grant no.\ 48138 from the John Templeton Foundation, 
EPSRC grant EP/L024233/1, and 
a Doctoral Training Grant from EPSRC (2nd author).
\newline
    A preliminary version of this article appeared in the proceedings
    of the conference ICALP'15 \cite{BCMS15}.
}

\maketitle

\section{Introduction}
\label{sec:intro}
The main aim in proof complexity is to understand the complexity of theorem proving. Arguably, what is even more important is to establish techniques for lower bounds, and the recent history of computational complexity speaks volumes on how difficult it is to develop general lower bound techniques. Understanding the size of proofs is important for at least two reasons. The first is its tight relation to the separation of complexity classes: NP vs.\ coNP for propositional proofs, and NP vs. PSPACE in the case of proof systems for quantified boolean formulas (QBF). New superpolynomial lower bounds for specific proof systems rule out specific classes of non-deterministic poly-time algorithms for problems in co-NP or PSPACE, thereby providing an orthogonal approach to the predominantly machine-oriented view of computational complexity. 

The second reason to study lower bounds for proofs is the analysis of
SAT and QBF solvers: powerful algorithms that efficiently solve the
classically hard problems of SAT and QBF for large classes of
practically relevant formulas. Modern SAT solvers routinely solve
industrial instances in even millions of variables for various
applications. Even though QBF solving is at a much earlier state, due
to its power to express problems more succinctly, QBF even applies to further fields such
as formal verification or
planning~\cite{Rin07,BM08a,EglyKLP14}. Each
successful run of a solver on an unsatisfiable instance can be
interpreted as a proof of unsatisfiability; and many modern SAT solvers
based on conflict-driven clause learning (CDCL) are known to
implicitly generate resolution proofs. 
Thus, understanding the complexity of resolution proofs helps obtain
worst-case  bounds for the performance of CDCL-based SAT solvers.

The picture is more complex for QBF solving, as there exist two main, yet conceptually very different paradigms: CDCL-based and expansion-based solving. A variety of QBF resolution systems have been designed to capture the power of QBF solvers based on these paradigms. The core system of these is Q-Resolution (\qrc), introduced by Kleine B\"{u}ning et al.\ \cite{DBLP:journals/iandc/BuningKF95}. This has been augmented to capture ideas from CDCL solving, leading to long-distance resolution (\lqrc) \cite{DBLP:journals/fmsd/BalabanovJ12}, universal resolution (\qurc) \cite{Gelder12}, or its combinations like \lquprc \cite{BWJ14}. 

Powerful proof systems for expansion-based solving were  developed in the form of \ecalculus \cite{JM15}, and the stronger \irc and \irmc \cite{BCJ14}. Recent findings show that CDCL and expansion are indeed orthogonal paradigms as the underlying proof systems from the two categories are incomparable with respect to simulations \cite{BCJ15}.

Understanding which general techniques can be used to show lower bounds for proof systems is of paramount importance in proof complexity. For propositional proof systems we have a number of very effective techniques, most notably the size-width technique  of Ben-Sasson and Wigderson \cite{BW01}, deriving size  from width bounds, game characterisations (e.g.\ \cite{Pud00,BK14}), the approach via proof-complexity generators (cf.\ \cite{Kra11}), and feasible interpolation. Feasible interpolation, first introduced by \Krajicek\ \cite{Kra97}, is a particularly successful paradigm that transfers circuit lower bounds to size of proof lower bounds. The technique has been shown to be effective for resolution \cite{Kra97}, cutting planes \cite{Pud97} and even strong Frege systems for modal and intuitionistic logics \cite{Hru09}. However, feasible interpolation fails for strong propositional systems as Frege systems under plausible cryptographic and number-theoretic assumptions \cite{KP98,BPR00,BDGMP04}.

The situation is drastically different for QBF proof systems, where we currently possess a very limited bag of techniques. 
In particular, the classical size-width technique of Ben-Sasson and Wigderson \cite{BW01} by which most resolution lower bounds are obtained drastically fails in \qrc \cite{BCMS16}.
At present we only have the recent strategy extraction technique  of Beyersdorff et al.\ \cite{BCJ15}, which works for \qrc as well as for stronger QBF Frege systems \cite{BBC16,BP16}, a game characterisation of the very weak tree-like \qrc \cite{BCS15}, and ad-hoc lower bound arguments for various systems \cite{BCJ15,DBLP:journals/iandc/BuningKF95}. In addition, Balabanov et al.\ \cite{BWJ14} develop methods to lift some previous lower bounds from \qrc to stronger systems.


We now proceed to explain the main contributions of the article.

\subsection*{1.\ A general lower bound technique.} We show that the feasible interpolation technique applies to all resolution-type QBF proof systems, whether expansion or CDCL based. This provides the first truly general lower bound technique for QBF proof systems, and---at the same time---hugely extends the scope of the feasible interpolation method. 
(We note that in recent work \cite{BCMS-FST16}, this technique has
also been shown to apply to a QBF version of the cutting planes proof system.)

In a nutshell, feasible interpolation works for true implications
$A(\vec{p},\vec{q}) \to B(\vec{p},\vec{r})$ (or, equivalently, false conjunctions $A(\vec{p},\vec{q}) \wedge \neg B(\vec{p},\vec{r})$), which by Craig's
interpolation theorem \cite{Cra57} possess interpolants $C(\vec{p})$
in the common variables $\vec{p}$. Such interpolants, even though they
exist, may not be of polynomial size \cite{Mun84}. However, it may be
the case that we can always efficiently extract such interpolants from
a proof of the implication in a particular proof system $P$, and in
this case, the system $P$ is said to admit feasible interpolation. If
we know that a particular class of formulas does not admit small
interpolants (either unconditional or under suitable assumptions),
then there cannot exist small proofs of the formulas in the system
$P$. Here we show that this feasible interpolation theorem holds for
arbitrarily quantified formulas $A(\vec{p},\vec{q})$ and
$B(\vec{p},\vec{r})$ above, when the common variables $\vec{p}$ are
existentially quantified before all other variables.

\subsection*{2.\ New lower bounds for QBF systems.} As our second
contribution we exhibit new hard formulas for QBF resolution
systems. Of course, exponential lower bounds for these
systems follow immediately  from the known lower bounds for resolution
(in these systems, refuting a totally quantified false sentence that
uses only existential quantifiers degenerates to classical
resolution). However, we can better understand the power of such systems to
handle arbitrary QBFs if we have more examples of false QBFs that use
existential and universal quantifiers in non-trivial ways and that are
hard to refute in these systems. 
It is fair to say that we are currently quite short of hard
examples: research so far has mainly concentrated on formulas of Kleine B\"{u}ning et al.\ \cite{DBLP:journals/iandc/BuningKF95} and their modifications \cite{BCJ15,BWJ14}, a principle by Janota and Marques-Silva \cite{JM15}, and a class of parity formulas recently introduced by Beyersdorff et al.\ \cite{BCJ15}. This again is in sharp contrast with classical proof complexity where a wealth of different combinatorial principles as well as  random formulas are known to be hard for resolution. 

Our new hard formulas are QBF contradictions formalising the easy and appealing fact that a graph cannot both have and not have a $k$-clique. The trick is that in our formulation, each interpolant for these formulas has to solve the $k$-clique problem. Using our interpolation theorem together with the exponential lower bound for the monotone circuit complexity of clique \cite{AB87}, we obtain exponential lower bounds for the clique-no-clique formulas in all CDCL and expansion-based QBF resolution systems. 

We remark that conceptually our clique-no-clique formulas are different from and indeed simpler than the clique-colour formulas used for the interpolation technique in classical proof systems. This is due to the more succinct expressivity of QBF. Indeed it is not clear how the clique-no-clique principle could even be formulated succinctly in propositional logic.

\subsection*{3.\ Comparison to strategy extraction.} On a conceptual level,
we uncover a tight  relationship between feasible interpolation and
strategy extraction. Strategy extraction is a very desirable property
of QBF proof systems and is known to hold for the main
resolution-based systems. From a refutation of a false QBF,
a winning strategy for the universal player can be efficiently extracted.

Like feasible interpolation, the lower bound technique based on
strategy extraction by Beyersdorff et al.\ \cite{BCJ15,BBC16} also transfers circuit lower bounds to proof size bounds. However, instead of monotone circuit bounds as in the case of feasible interpolation, the strategy extraction technique imports $\mathsf{AC}^0$ circuit lower bounds (or further circuit bounds for circuit classes directly corresponding to the lines in the proof system \cite{BBC16}). Here we show that each feasible interpolation problem can be transformed into a strategy extraction problem, where the interpolant corresponds to the winning strategy of the universal player on the first universal variable. This clarifies that indeed feasible interpolation can be viewed as a special case of strategy extraction.

\subsection*{Organisation of the paper}

The remaining part of this article is organised as follows. In Section~\ref{sec:prelim} we review the definitions and relations of relevant QBF proof systems. In Section~\ref{sec:interpolation} we start by recalling the overall idea for feasible interpolation and show  interpolation theorems for the strongest CDCL-based system \lquprc as well as the strongest expansion-based proof system \irmc. This implies feasible interpolation for all QBF resolution-based systems. Further we show that all these systems even admit monotone feasible interpolation. In Section~\ref{sec:lower-bounds} we obtain the new lower bounds for the clique-no-clique formulas. Section~\ref{sec:strat-extraction} reformulates interpolation as a strategy extraction problem.

\section{Preliminaries} 
\label{sec:prelim}

A literal is a boolean variable or its negation. We say a literal $x$ is complementary to the literal $\neg x$ and vice versa. 
A {\em clause} is a disjunction of literals. For notational
convenience, we sometimes also refer
to a clause as a set of literals.
The empty clause is denoted by~$\Box$, and is semantically equivalent
to false. We denote true by 1 and false by 0. 
A formula in {\em conjunctive normal form} (CNF) is a
conjunction of clauses.  
For a literal $l=x$ or $l=\lnot x$, we write $\var(l)$ for~$x$ and extend this notation to $\var(C)$ for a clause $C$.
 Let $\alpha$ be any partial assignment. For a clause $C$, we write
 $C|_{\alpha}$ for the clause obtained after applying the partial
 assignment $\alpha$ to $C$. For example, applying $\alpha:~x_1 \leftarrow 0$ to the clause $C \equiv (x_1 \vee x_2 \vee x_3)$ 
yields $C|_{\alpha} \equiv (x_2 \vee x_3)$, and applying $\alpha':~x_1 \leftarrow 1$ to the same clause gives $C|_{\alpha'} \equiv 1$. 
In the former case, we say that $C$ evaluates to the clause $(x_2 \vee x_3)$ under the assignment $\alpha$, and in the latter case, it 
evaluates to $1$ under the assignment $\alpha'$.
Similarly, for a formula $F$, we write
 $F|_{\alpha}$ for the restriction of the formula to the partial
 assignment. 

Quantified Boolean Formulas (QBFs) extend propositional logic with
the boolean quantifiers $\forall$ and $\exists$. They have the 
standard semantics that $\forall x. F$ is 
satisfied by the same truth assignments to its free variables as
$F|_{x = 0} \wedge F|_{x = 1}$, and $\exists x. F$ as $F|_{x = 0} \vee
F|_{x = 1}$. 
We assume that QBFs are fully quantified (no free variables), 
in \emph{closed prenex form}, and  with a CNF
matrix, i.e, we consider the form $Q_1 x_1 \dots
Q_n x_n. \phi$, where $Q_i \in \{\exists, \forall \}$, and 
the formula $\phi$  is in CNF and is defined  on the set of variables  
$X = \{x_1, \ldots , x_n\}$. Further, we assume that complementary
literals do not appear in the same clause; that is, no clause in the
matrix is tautological.  
 The propositional part $\phi$ is called the {\em matrix} and 
the rest the {\em prefix}. 
We abbreviate the prefix by the notation $\mathcal{Q} x$. 
The \emph{index} $\ind(x)$ of a variable is its position in the
prefix; thus $\ind(x_i)=i$. When $Q_i=\exists$ ($Q_i=\forall$,
respectively), we say that $x_i$ is an existential variable (a
universal variable, resp.\ ).   A literal $l$ is said to be
existential (universal) if $\var(l)$ is existential (universal,
resp.\ ). For a literal $l$, we write $\ind(l)$ for $\ind(\var(l))$.

Often it is useful to think of a QBF $\mathcal{Q} x .\, \phi$
 as a \emph{game} between the \emph{universal} and the
 \emph{existential player}. 
 In the $i$-th step of the game, player $Q_i$ assigns a value to the
 variable $x_i$. 
 The existential
 player wins the game iff the matrix~$\phi$ evaluates to $1$ under
 the assignment constructed in the game. The universal player wins
 iff the matrix~$\phi$ evaluates to $0$.
 Let $u$ be a universal variable~$u$ with index~$i$. At the $i$th step
 of the game, when the universal player has to decide what value to
 assign to $u$, all variables with index less than $i$ already have
 values assigned to them. A \emph{strategy for $u$} is thus a function
 from the set of assignments to the variables with index $<i$ to
 $\{0,1\}$.  A strategy for the universal player is a collection of strategies, one for each universal
 variable. A strategy  is a \emph{winning strategy} for the universal player if,
 using it, the universal player can win any possible game, 
 irrespective of the strategy used by the existential
 player. 
A QBF is false iff there exists a \emph{winning strategy}
 for the universal player
~(\cite{Goultiaeva-ijcai11}, \cite[Sec.\,4.2.2]{AroraBarak09},
 \cite[Chap.\,19]{Pap94}).


\smallskip
\textbf{Resolution-based calculi for QBF.}
We now give a brief overview of the main existing resolution-based
calculi for QBF. For the technical proofs in this paper, full details
are needed only for two systems, \lquprc\ and \irmc, both
of which are included in the overview. 

Recall that resolution for propositional proofs (where all variables
are existential) operates by inferring clauses, starting from the
clauses of the given formula (axioms), until the empty clause is
derived. From clauses $C\vee x$ and $D\vee \neg x$ that have been
already inferred, it can infer the clause $C \vee D$, by resolving on
the variable $x$. Here, $x$ is referred to as the pivot, and $C\vee D$
is the resolvent. The clauses $C\vee x$ and $D\vee \neg x$ are
referred to as parents of the clause $C\vee D$ in the proof. In a
representation of a proof as a graph, each clause in the proof is a
node, and edges are directed from parent to child.
This system can be augmented in various ways to
handle QBFs with universal variables.

We start by describing the proof systems modelling 
\emph{CDCL-based QBF solving};
their rules are summarized in Figure~\ref{fig:allrules}. The most basic and important 
system is \emph{Q-resolution (\qrc)} by Kleine B\"{u}ning et al.\
\cite{DBLP:journals/iandc/BuningKF95}. It is a resolution-like
calculus 
that operates on QBFs in prenex form with CNF matrix. The lines in a \qrc proof are clauses. In addition to the axioms,
\qrc comprises the resolution rule S$\exists$R and universal reduction $\forall$-Red (cf.\ Figure~\ref{fig:allrules}).
Note that the conditions in S$\exists$R explicitly disallow inferring
a tautology; this is syntactically essential for soundness. 

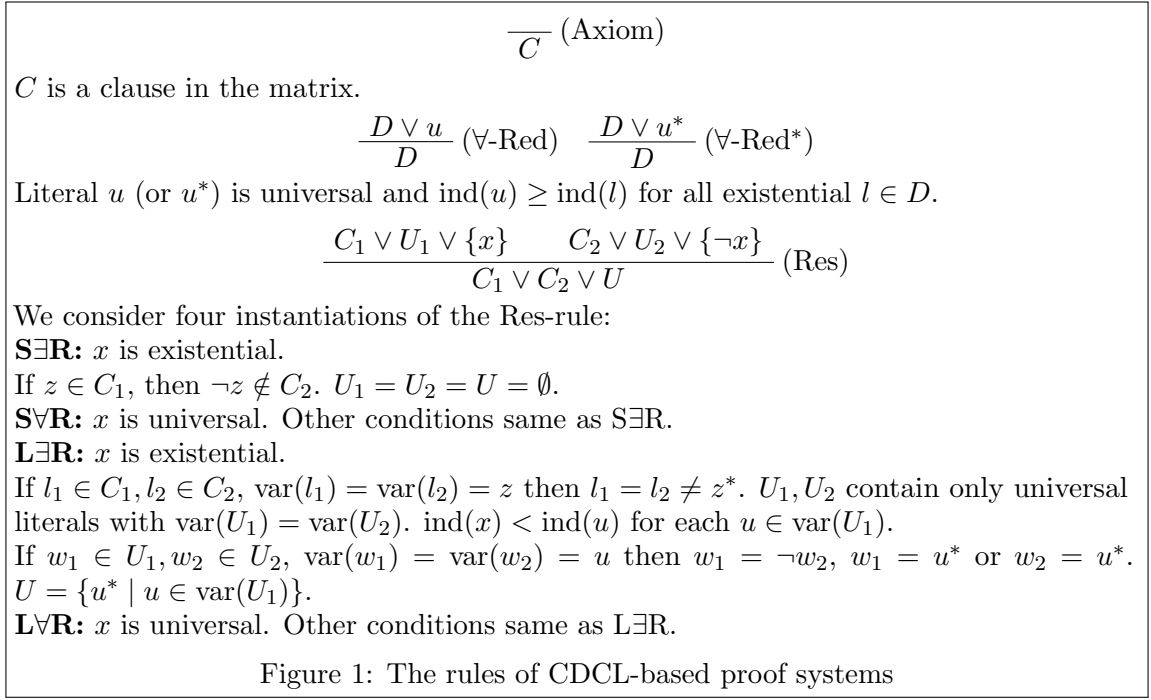
\begin{figure}[h!]
\framebox{\parbox{\breite}
  {
\begin{prooftree}
\AxiomC{}
\RightLabel{(Axiom)}
\UnaryInfC{$C$}
\end{prooftree}
\begin{minipage}{0.99\linewidth}
$C$ is a clause in the matrix. 
\end{minipage}

\begin{prooftree}
\AxiomC{$D\vee u$}
\RightLabel{($\forall$-Red)}
\UnaryInfC{$D$}
\DisplayProof\hspace{0cm}
\AxiomC{$D\vee u^*$}
\RightLabel{($\forall$-Red$^*$)}
\UnaryInfC{$D$}
\end{prooftree}
\begin{minipage}{0.99\linewidth}
Literal $u$ (or $u^*$) is universal and
  $\ind(u)\geq\ind(l)$ for all existential $l\in D$.
\end{minipage}

\begin{prooftree}
\AxiomC{$C_1\vee U_1\vee\{x\}$}
\AxiomC{$C_2 \vee U_2\vee\{\lnot{x}\}$}
\RightLabel{(Res)}
\BinaryInfC{$C_1\vee C_2\vee U$}
\end{prooftree}
\begin{minipage}{0.99\linewidth}
We consider four instantiations of the Res-rule:\\
\textbf{S$\exists$R:}  $x$ is existential.\\ If $z\in C_1$, then $\lnot{z}\notin C_2$. $U_1=U_2=U=\emptyset$.\\
\textbf{S$\forall$R:}  $x$ is universal. Other conditions same as S$\exists$R.\\
\textbf{L$\exists$R:}  $x$ is existential. \\If  $l_1\in C_1, l_2\in C_2$, $\var(l_1)=\var(l_2)=z$ then $l_1=l_2\neq z^*$.
 $U_1, U_2$ contain only universal literals with $\var(U_1)=\var(U_2)$.
 $\ind(x)<\ind(u)$ for each $u\in\var(U_1)$.\\
 If $w_1\in U_1, w_2\in U_2$, $\var(w_1)=\var(w_2)=u$ then $w_1=\lnot w_2$, $w_1=u^*$ or $w_2=u^*$. $U=\{u^* \mid u\in \var(U_1)\}$.\\
\textbf{L$\forall$R:}  $x$ is universal. Other conditions same as L$\exists$R.
\end{minipage}
\caption{The rules of CDCL-based proof systems}
\label{fig:allrules}
}}
\end{figure}
\emph{Long-distance resolution (\lqrc)} appears originally in the work of Zhang and Malik \cite{DBLP:conf/iccad/ZhangM02}
and was formalized into a calculus by Balabanov and Jiang \cite{DBLP:journals/fmsd/BalabanovJ12}.
It allows resolving clauses $C\vee x$ and $D\vee \neg x$ on an
existential variable $x$ even if $C$ and $D$ contain complementary
literals (where $C\vee D$ would be a tautology), provided the
complementary literals correspond to universal variables with index  greater than the index of the pivot variable $x$. 
It merges complementary literals of a universal variable~$u$
into the special literal~$u^*$ which then appears in the resolvent. We define $\ind(u^*) = \ind(u)$. 
\lqrc uses the rules  L$\exists$R, $\forall$-Red and $\forall$-Red$^*$ (cf.\ Figure~\ref{fig:allrules}).

\emph{QU-resolution (\qurc)} by Van Gelder \cite{Gelder12} removes the restriction from \qrc that the resolved variable must be an existential variable and allows resolution of universal variables. The rules of \qurc are  S$\exists$R, S$\forall$R and $\forall$-Red (cf.\ Figure~\ref{fig:allrules}). 
%

\emph{\lquprc} by Balabanov et al.\ \cite{BWJ14} extends \lqrc by allowing short and long distance resolution pivots to be universal. However, the pivot is never a merged literal $z^*$. \lquprc uses the rules L$\exists$R, L$\forall$R, $\forall$-Red and $\forall$-Red$^*$ (cf.\ Figure~\ref{fig:allrules}). 

The second type of calculi models \emph{expansion-based QBF solving}. These calculi are 
based on \emph{instantiation} of universal variables:  
\ecalculus by Janota and Marques-Silva \cite{JM15}, \irc, and \irmc by Beyersdorff et al.\ \cite{BCJ14}.  All these
calculi operate on clauses that comprise only existential variables from the original QBF, which
are additionally \emph{annotated} by a substitution to some universal variables, e.g.\ $\lnot
x^{0/u_1 1/u_2}$.
For any annotated literal $l^\sigma$, the substitution $\sigma$ must not make
assignments to variables at a higher quantification level than $l$, i.e.\ if
$u\in\domain(\sigma)$, then $u$ is universal and $\ind(u)<\ind(l)$.  
To preserve this invariant, we use the \emph{auxiliary notation $l^{[\sigma]}$}, which for an existential literal $l$  and an assignment $\sigma$ to the universal variables 
filters out all assignments that are not permitted,
i.e.\ $l^{[\sigma]}=l^{\comprehension{c/u\in\sigma}{\ind(u)<\ind(l), ~c \in \{0, 1 \} }}$.

As annotations can be partial assignments, we use auxiliary operations of \emph{completion} and \emph{instantiation}.  For assignments
$\tau$ and $\mu$, we write $\complete{\tau}{\mu}$ for the assignment $\sigma$ defined as follows:
$\sigma(x)=\tau(x)$ if $x\in\domain(\tau)$,  
otherwise $\sigma(x)=\mu(x)$ if $x\in\domain(\mu)\backslash\domain(\tau)$.
The operation $\complete{\tau}{\mu}$ is called \emph{completion} because $\mu$
provides values for variables not defined in $\tau$.  
The operation is associative and therefore
we can omit parentheses.  For an assignment $\tau$ and
an annotated clause $C$, the function $\instantiate(\tau,C)$ returns the annotated clause
$\comprehension{l^{[\complete{\sigma}{\tau}]}}{l^\sigma\in C}$.  The system \irc is
defined in Figure~\ref{fig:IRC}.  

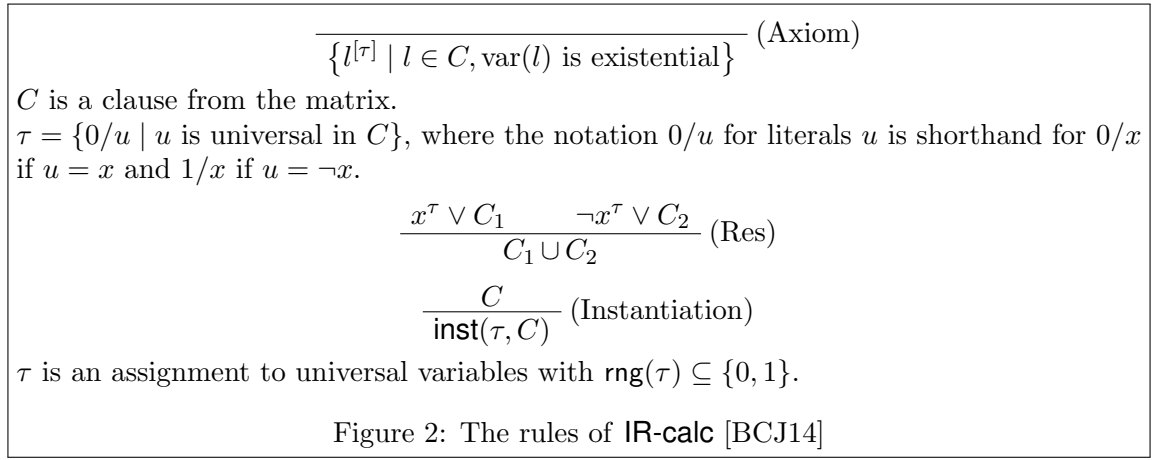
\begin{figure}[h!]
  \framebox{\parbox{\breite}
  {
    \begin{prooftree}
      \AxiomC{}
      \RightLabel{(Axiom)}
      \UnaryInfC{$\comprehension{l^{[\tau]}}{l\in C, \var(l)\text{ is existential}} $}
    \end{prooftree}
    $C$ is a  clause from the matrix. \\$\tau=\comprehension{0/u}{u\text{ is universal in }C}$, where the notation $0/u$ for literals $u$ is shorthand for $0/x$ if $u=x$ and $1/x$ if $u=\neg x$.
    \begin{prooftree}
      \AxiomC{$x^\tau\lor C_1 $ }
      \AxiomC{$\lnot x^\tau\lor C_2 $}
      \RightLabel{(Res)}
      \BinaryInfC{$C_1\union C_2$}
 \end{prooftree}
    \begin{prooftree}
      \AxiomC{$C$}
      \RightLabel{(Instantiation)}
      \UnaryInfC{$\instantiate(\tau,C)$}
    \end{prooftree}
    $\tau$ is an assignment to universal variables with $\range(\tau) \subseteq \{0,1\}$.
  \caption{The rules of \irc  \cite{BCJ14}}\label{fig:IRC}
}}
\end{figure}

The calculus \irmc further extends \irc by enabling annotations containing $*$.
The rules of the calculus \irmc are presented in Figure~\ref{fig:IRMC}.
The symbol $*$  may be introduced by the merge rule, e.g.\ by
collapsing $x^{0/u}\lor x^{1/u}$ into $x^{*/u}$.  

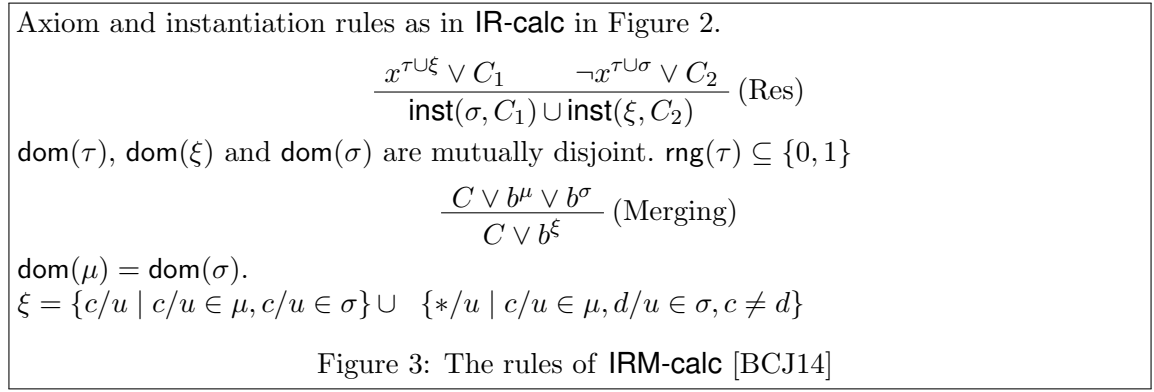
\begin{figure}[h!]
  \framebox{\parbox{\breite}
  {
    Axiom and instantiation rules as in \irc in Figure \ref{fig:IRC}.
    \begin{prooftree}
      \AxiomC{$x^{\tau\cup\xi}\lor C_1 $ }
      \AxiomC{$\lnot x^{\tau\cup\sigma}\lor C_2 $}
      \RightLabel{(Res)}
      \BinaryInfC{$\instantiate(\sigma,C_1)\union\instantiate(\xi,C_2)$}
    \end{prooftree}
    \text{$\domain(\tau)$, $\domain(\xi)$ and $\domain(\sigma)$ are
   mutually disjoint.} $\range(\tau)\subseteq\{0,1\}$
    \begin{prooftree}
      \AxiomC{$C\lor b^\mu\lor b^\sigma$}
      \RightLabel{(Merging)}
      \UnaryInfC{$C\lor b^\xi$}
    \end{prooftree}
    \text{$\domain(\mu)=\domain(\sigma)$.}

    \text{$\xi=\comprehension{c/u}{c/u\in\mu,c/u\in\sigma}\union$}
    \text{\hphantom{$\xi$}$\comprehension{*/u}{c/u\in\mu,d/u\in\sigma,c\neq d}$ }
  \caption{The rules of \irmc \cite{BCJ14} }\label{fig:IRMC}
}}
\end{figure}

The simulation order of QBF resolution systems is shown in Figure~\ref{fig:sim-structure}. All proof systems have been exponentially separated (cf.\ \cite{BCJ15}).

\begin{figure}[h!]
  \parbox{\breite}{
\centering
\begin{tikzpicture}[scale=1.1]
\node[calcn](n1) at (3,1) {{\sf Tree-}\qrc} ;
\node[calcn](n2) at (3,2) {\qrc} ;
\node[expcalcn](n3) at (0,2) {\ecalculus} ;
\node[calcn](n4) at (2,3) {\lqrc} ;
\node[calcn](n5) at (4,3) {\qurc} ;
\node[calcn](n6) at (3,4) {\lquprc} ;
\node[expcalcn](n8) at (0,3) {\irc} ;
\node[expcalcn](n9) at (0,4) {\irmc} ;
\draw(n4)--(n2);
\draw(n6)--(n4);
\draw(n2)--(n8);
\draw(n8)--(n3);
\draw (n3)--(n1)--(n2)--(n5);
\draw(n5)--(n6);
\draw(n8)--(n9)--(n4);
\end{tikzpicture}
}
  \caption{The simulation order of QBF resolution systems. Systems on the left correspond to expansion-based solving, whereas the systems on the right are CDCL based.}\label{fig:sim-structure}
\end{figure}

We end this section with a  definition that is used in later
sections. It generalises the notion of weakening. 
For clauses containing only literals of the
form $x_i, \neg x_i$ (no $l^*$), clause $D$ weakens clause $C$ if
every literal in $C$ is also present in $D$; i.e.\ $C \subseteq D$. 
With merged literals, the analogous notion of weakening is as defined below. 

\begin{definition}\label{def:preceq}
For clauses $C,D$ we write $C\preceq D$ if for any literal $l\in C$ we have $l\in D$ or $l^*\in D$ and for any $l^*\in C$ we have $l^*\in D$.

For annotations $\tau$ and $\sigma$ we say that $\tau\preceq\sigma$ if $\domain(\tau)=\domain(\sigma)$ and for any $c/u\in \tau$ we have  $c/u\in \sigma$ or $*/u\in \sigma$ and for any  $*/u\in \tau$ we have $*/u\in \sigma$. 
If $C,D$ are annotated clauses, we write  $C\preceq D$ if there is an injective function $f:C \hookrightarrow D$ such that for all $l^\tau\in C$ we have  $f(l^\tau) =l^\sigma$ with  $\tau\preceq\sigma$.
\end{definition}
Note: the requirement above that $f$  is injective ensures that
$x^{0/u}\vee x^{1/v} \not\preceq x^{0/u,1/v}$.

\section{Feasible Interpolation and Feasible Monotone Interpolation}
\label{sec:interpolation}

In this section we show that feasible interpolation and feasible
monotone interpolation hold for \lquprc and \irmc. We adapt the
technique first used by \Pudlak \cite{Pud97} to re-prove and
generalise the result of \Krajicek \cite{Kra97}.

\subsection{The setting}\label{subsec:interpolation-setting}
Consider a false QBF  $\mathcal{F}$ of the form
$$\exists \vec{p} \mathcal{Q} \vec{q} \mathcal{Q} \vec{r} \big[
  A(\vec{p}, \vec{q}) \wedge B(\vec{p}, \vec{r})\big], 
$$
where, $\vec{p}$, $\vec{q}$, and $\vec{r}$ are mutually disjoint sets of
propositional variables, $A(\vec{p}, \vec{q})$ is a CNF formula on
variables $\vec{p}$ and $\vec{q}$, and $B(\vec{p}, \vec{r})$ is a CNF
formula on variables $\vec{p}$ and $\vec{r}$. 
Thus $\vec{p}$ contains all the 
common variables between them. The $\vec{q}$ and
$\vec{r}$ variables can be quantified arbitrarily, with any number of
alternations between quantifiers. The QBF is equivalent to the
following, not in prenex form
$$\exists \vec{p} \big[ \mathcal{Q} \vec{q}. A(\vec{p}, \vec{q}) \wedge \mathcal{Q} \vec{r}. B(\vec{p}, \vec{r})\big].$$

Let $\vec{a}$ denote an assignment to the $\vec{p}$ variables. We
denote $A(\vec{p},\vec{q})|_{\vec{a}}$ 
by $A(\vec{a},\vec{q})$ and $B(\vec{p},\vec{q})|_{\vec{a}}$ by
$B(\vec{a},\vec{q})$.

\begin{definition}
Let $\mathcal{F}$ be a false QBF of the form $\exists \vec{p} \mathcal{Q} \vec{q} \mathcal{Q} \vec{r}. \left[
  A(\vec{p}, \vec{q}) \wedge B(\vec{p}, \vec{r})\right]$.
An \emph{interpolation circuit} for $\mathcal{F}$ is a boolean circuit $G$
such that on every $0, 1$ assignment $\vec{a}$ for $\vec{p}$ we have
\begin{align*} 
G(\vec{a}) = 0 &\implies \mathcal{Q} \vec{q}. A(\vec{a}, \vec{q}) \text{ is false, and }\\  
G(\vec{a}) = 1 &\implies \mathcal{Q} \vec{r}. B(\vec{a}, \vec{r}) \text{ is false. }
\end{align*}
We say that a QBF proof system $S$ has \emph{feasible interpolation}
if there is an effective procedure that, given any 
$S$-proof $\pi$ of a QBF $\mathcal{F}$ of the form above, 
outputs an interpolation circuit for $\mathcal{F}$ of size
polynomial in the size of $\pi$.

We say that the procedure \emph{extracts} the circuit from the proof.

We say that $S$ has \emph{monotone feasible interpolation} if the
following holds: in the same setting as above, if $\vec{p}$ appears
only positively in $A(\vec{p}, \vec{q})$, then the  
interpolation circuit for $\mathcal{F}$ extracted from $\pi$ is monotone.
\end{definition}

As our main results, we show that both \lquprc and \irmc have monotone feasible interpolation.

Before proving the interpolation theorems, we first outline the general idea:

\subsubsection*{Proof idea} Fix a proof system $S\in \{$\lquprc, \irmc{}\,$\}$ and an $S$-proof $\pi$ of $\mathcal{F}$.
Consider the following definition of a $\vec{q}$-clause and an $\vec{r}$-clause.
\begin{definition}
We call a clause $C$ in $\pi$ a $\vec{q}$-clause
(resp.\ $\vec{r}$-clause), if $C$ contains only variables $\vec{p},
\vec{q}$ (resp. $\vec{p}, \vec{r}$). We also call $C$ a
$\vec{q}$-clause (resp.\ $\vec{r}$-clause), if $C$ contains only
$\vec{p}$ variables, but all its descendant clauses in the proof $\pi$
(all clauses with a directed path to $C$ in $\pi$) are $\vec{q}$
(resp.\ $\vec{r}$)-clauses. In the case of \irmc the variables
appearing in the annotations are irrelevant and can be from either set.
\end{definition}
From $\pi$ we construct a circuit $C_\pi$ with the $\vec{p}$-variables
as inputs: for each node $u$ with clause $C_u$ in the proof $\pi$,
associate a gate $g_u$ (or a constant-size circuit) in the circuit
$C_\pi$. Next, we inductively construct, for any assignment $\vec{a}$ to the
$\vec{p}$ variables, another proof-like structure $\pi'(\vec{a})$. For
each node $u$ with clause $C_u$ in the proof $\pi$, associate a clause
$C'_{u, \vec{a}}$ in the structure $\pi'(\vec{a})$.  Finally, we
obtain $\pi''(\vec{a})$ from the structure $\pi'(\vec{a})$ by
instantiating $\vec{p}$ variables to the assignment $\vec{a}$ (that
is, $C''_{u,\vec{a}} = C'_{u,\vec{a}}|_{\vec{a}}$ for each node $u$) and
doing some pruning, and show that $\pi''(\vec{a})$ is a valid proof in
$S$. We then find that if $C_\pi(\vec{a})=0$, then $\pi''(\vec{a})$
uses only $\vec{q}$-clauses and thus is a refutation of $\mathcal{Q}
\vec{q}. A(\vec{a}, \vec{q})$, and if $C_\pi(\vec{a})=1$, then
$\pi''(\vec{a})$ uses only $\vec{r}$-clauses and thus is a refutation
of $\mathcal{Q} \vec{r}. B(\vec{a}, \vec{r})$. Thus $C_\pi$ is the
desired interpolant circuit.

More precisely, we show by induction on the height of $u$ in $\pi$
(that is, the length of the longest path to $u$ from a source node in
$\pi$) that: 
\begin{enumerate}
	\item $C'_{u,\vec{a}} \preceq C_u$.

	\item $g_{u}(\vec{a}) = 0 \implies C''_{u,\vec{a}}$ is a $\vec{q}$-clause and can be obtained from the clauses of 
			$A(\vec{a},\vec{q})$ alone using the rules of $S$.

	\item $g_{u}(\vec{a}) = 1 \implies C''_{u,\vec{a}}$ is an $\vec{r}$-clause and can be obtained from the clauses of
			$B(\vec{a},\vec{r})$ alone using the rules of $S$.
\end{enumerate} 

From the above, we have the following conclusion. Let $r$ be the root of $\pi$. Then on any assignment $\vec{a}$ to the $\vec{p}$ variables we have:

\begin{enumerate}
	\item $C'_{r,\vec{a}} \preceq C_r = \Box$, so $C'_{r,\vec{a}} = \Box$. Therefore, $C''_{r,\vec{a}} = C'_{r,\vec{a}}|_{\vec{a}} = \Box$.

	\item $g_{r}(\vec{a}) = 0 \implies  \Box$ is a $\vec{q}$-clause and can be obtained from the clauses 			of $A(\vec{a},\vec{q})$ alone using the rules of system $S$. Hence by soundness of $S$, 
		$\mathcal{Q} \vec{q}. A(\vec{a}, \vec{q})$ is false.

	\item $g_{r}(\vec{a}) = 1 \implies  \Box$ is an $\vec{r}$-clause and can be obtained from the clauses 			of $B(\vec{a},\vec{r})$ alone using the rules of system $S$. Hence by soundness of $S$,
		$\mathcal{Q} \vec{r}. B(\vec{a}, \vec{r})$ is false.
\end{enumerate}
\noindent
Thus $g_r$, the output gate of the circuit, computes an interpolant. 

When $\mathcal{F}$ has only existential quantification, $\pi$ is a
classical resolution proof, and this is exactly the interpolant computed
by \Pudlak's method \cite{Pud97}. The challenge here is to construct $\pi'$
and $\pi''$ appropriately when the stronger proof systems are used for
general QBF, while maintaining the inductive invariants. 


\subsection{Interpolants from \lquprc proofs}
\label{subsec:lpq-interpolant}
We now implement the idea described above for \lquprc.

\begin{theorem} \label{thm:lqup}
   \lquprc has feasible interpolation.
\end{theorem}
\begin{proof}
As mentioned in the proof idea, for an \lquprc proof $\pi$ of $\mathcal{F}$  we first describe the circuit $C_{\pi}$ with input $\vec{p}$.

\medskip\noindent
{\bf Construction of the circuit $C_{\pi}$:} The DAG underlying the
circuit is exactly the same as the DAG underlying the proof $\pi$. 
For each node $u$ with clause $C_u$ in $\pi$ we associate a gate $g_u$ as follows:

\begin{description}

	\item[$u$ is a leaf node: ] If $C_u \in A(\vec{p}, \vec{q})$ then $g_u$ is a constant $0$ gate. 
				    If $C_u \in B(\vec{p}, \vec{r})$ then $g_u$ is a constant $1$ gate.
\end{description}

\noindent
{\bf $u$ is an internal node:} We distinguish four cases. 
		
\begin{enumerate}
			
\item $u$ was derived by 
a universal reduction step. 
In this case put a no-operation gate (identity gate) for $g_u$.

\item $u$ corresponds to a resolution step with an existential
  variable $x \in \vec{p}$ as pivot.  
Nodes $v$ and $w$ are its two parents, i.e.
$$\frac{\overbrace{C_1 \vee x}^\text{node $v$} \hspace{7mm} \overbrace{C_2\vee \neg x}^\text{node $w$}}{\underbrace{ C }_\text{node $u$} }$$
In this case, put a selector gate
$\sel(x, g_v, g_w)$ for $g_u$. Here,
$\sel(x, a, b) = a$, when $x = 0$ and
$\sel(x, a, b) = b$, when $x =
1$. That is,
$\sel(x, a, b) = (\neg x \wedge a) \vee
(x \wedge b)$. 
Note that all the variables in
$\vec{p}$ are existential variables
without annotations (equivalently, with empty annotations). 

\item $u$ corresponds to a resolution step with an existential or universal variable $x \in \vec{q}$ as pivot. Put an OR gate for $g_u$.

\item $u$ corresponds to a resolution step with an existential or universal variable $x \in \vec{r}$ as pivot. Put an AND gate for $g_u$.
\end{enumerate}
This completes the description of the circuit $C_{\pi}$.

\medskip\noindent
{\bf Construction of $\pi'$ and $\pi''$:} 
Following our proof idea, we now describe, for each node $u$ in $\pi$
with clause $C_u$, the associated clause $C'_{u,\vec{a}}$ in
$\pi'(\vec{a})$. Once $\pi'(\vec{a})$ is defined, 
the structure $\pi''(\vec{a})$ is
obtained by instantiating $\vec{p}$ variables by
the assignment $\vec{a}$ in each clause of $\pi'(\vec{a})$, cutting
away any edge out of a node 
where the clause evaluates to $1$, and deleting nodes which now have
no path to the root node.
That is, for each  node $u$, if $C'_{u,\vec{a}}|_{\vec{a}} = 1$, then
the node $u$ is removed,  and otherwise the node $u$ survives and 
 the associated clause $C''_{u,\vec{a}}$ is equal 
to $C'_{u,\vec{a}}|_{\vec{a}}$.

We show (by induction on the height of $u$ in $\pi$) that: 

\begin{enumerate}
	\item $C'_{u,\vec{a}} \preceq C_u$.

	\item $g_{u}(\vec{a}) = 0 \implies C''_{u,\vec{a}}$ is a $\vec{q}$-clause and can be obtained from the clauses of 
			$A(\vec{a},\vec{q})$ alone using the rules of system \lquprc.

	\item $g_{u}(\vec{a}) = 1 \implies C''_{u,\vec{a}}$ is a $\vec{r}$-clause and can be obtained from the clauses of
			$B(\vec{a},\vec{r})$ alone using the rules of system \lquprc.
\end{enumerate}
As described in the proof outline, this suffices to conclude that
$C_\pi$ computes an interpolant.
We now present the construction details.

\medskip\noindent
{\bf At leaf level:} Let node $u$ be a leaf in $\pi$. Then
$C'_{u,\vec{a}} = C_u$; that is, we copy the clause as it is. Trivially,
we have $C'_{u,\vec{a}} \preceq C_u$. By construction of $C_\pi$, the
conditions concerning $g_u(\vec{a})$ and $C''_{u,\vec{a}}$ are satisfied.

\medskip\noindent
At an internal node we distinguish four cases based on the rule that was applied.
	
\medskip\noindent
{\bf At an internal node with universal reduction:} 
Let node $u$ be an internal node in $\pi$ corresponding to a universal reduction step on some universal literal
$x$ or $x^*$. Let node $v$ be its only parent. Here we consider only the case where the  universal literal is $x$. The case of $x^*$ is identical. We have 
$$\frac{C_v =\overbrace{D_v \vee x }^\text{node $v$}}{C_u =
  \underbrace{D_v}_\text{node $u$}}, \qquad x \text{ is a universal
  literal, } \forall \text{~existential literal~}  l \in D_v, \ind(l) < \ind(x).$$

In this case, define $C'_{u,\vec{a}} = C'_{v,\vec{a}} \setminus \{ x, \neg x, x^* \}$. 
By induction, $C'_{v,\vec{a}} \preceq C_v = D_v \vee x $. Therefore, $C'_{u,\vec{a}} = C'_{v,\vec{a}} \setminus \{ x , \neg x, x^*\}  \preceq D_v = C_u$.

If $g_u(\vec{a}) = 0 $, then we know that $g_v(\vec{a}) = 0 $ as
$g_u(\vec{a}) = g_v(\vec{a})$. By the induction hypothesis, we know
that $C''_{v,\vec{a}} = C'_{v,\vec{a}}|_{\vec{a}}$ is a $\vec{q}$-clause and can be derived using
$A(\vec{a}, \vec{q})$ alone via \lquprc. Recall that $C'_{u,\vec{a}} = C'_{v,\vec{a}} \setminus \{x, \neg x, x^*\}$
in this case. Since $\vec{a}$ is an assignment to the $\vec{p}$
variables and $x \notin \vec{p}$, $C'_{u,\vec{a}}|_{\vec{a}} =
C''_{u,\vec{a}}$ is also a $\vec{q}$-clause and can be derived using
$A(\vec{a}, \vec{q})$ alone via \lquprc. (Either
$C''_{u,\vec{a}}$ already equals $C''_{v,\vec{a}}$, or $x$ needs to be dropped. In the
latter case,
the condition on $\ind(x)$
is satisfied at $C''_{u,\vec{a}}$ because it is satisfied
at $C_v$ in $\pi$ and $C'_{v,\vec{a}} \preceq C_v$.
So we can drop $x$ from $C''_{v,\vec{a}}$ to get $C''_{u,\vec{a}}$.)

The situation is dual for the case when $g_u(\vec{a})=1$; we get
$\vec{r}$-clauses.

\medskip\noindent
{\bf At an internal node with $\vec{p}$-resolution:}   
Let node $u$ in the proof $\pi$ correspond to a
          resolution step with pivot $x \in \vec{p}$. Note that
          $x$ is existential, as $\vec{p}$ variables occur only existentially in $\mathcal{F}$. We have
$$
\frac{C_v = \overbrace{C_1 \vee U_1 \vee x}^\text{node $v$} \hspace{7mm} \overbrace{C_2 \vee U_2 \vee \neg x}^\text{node $w$} = C_w}{C_u = \underbrace{C_1 \vee C_2 \vee U}_\text{node $u$}}. 
$$
In the assignment $\vec{a}$, if $x = 0$, then define
$C'_{u,\vec{a}} = C'_{v,\vec{a}} \setminus \{x\} $ and if $x = 1$ then define
$C'_{u,\vec{a}} = C'_{w,\vec{a}} \setminus \{ \neg x\}$. By induction, we have
$C'_{v,\vec{a}} \preceq C_v$ and $C'_{w,\vec{a}} \preceq
C_w$. 
So, if $x = 0$, we have $C'_{u,\vec{a}} = C'_{v,\vec{a}} \setminus \{x\} \preceq C_1 \vee U_1 \preceq C_u$.
If $x = 1$, we have $C'_{u,\vec{a}} \preceq C'_{w,\vec{a}} \setminus \{ \neg x\} \preceq C_2 \vee U_2 \preceq C_u$.


In this case $g_u$ is a selector gate. If $x = 0$ in the assignment
$\vec{a}$, then $g_u(\vec{a}) = g_v(\vec{a})$ and $C''_{u,\vec{a}} =
C''_{v,\vec{a}}$. Since the conditions concerning  $g_v(\vec{a})$  and
$C''_{v,\vec{a}}$ are satisfied by induction,  the conditions concerning  $g_u(\vec{a})$  and
$C''_{u,\vec{a}}$ are satisfied as well. Similarly, if $x=1$, then  
$g_u(\vec{a}) = g_w(\vec{a})$ and $C''_{u,\vec{a}} =
C''_{w,\vec{a}}$, and the statements that are inductively true at $w$
hold at $u$ as well. 

\medskip\noindent
{\bf At an internal node with $\vec{q}$-resolution:} 
  Let node $u$ in the proof $\pi$ correspond to a
          resolution step with pivot $x \in \vec{q}$. Note that $x$ may be existential or universal. We have
$$
  \frac{C_v = \overbrace{C_1 \vee U_1 \vee x}^\text{node $v$} \hspace{7mm} \overbrace{C_2 \vee U_2 \vee \neg x}^\text{node $w$} = C_w}{C_u = \underbrace{C_1 \vee C_2 \vee U}_\text{node $u$}}, \quad x \in \vec{q}.
$$

If $g_v(\vec{a}) = 1$ then define $C'_{u,\vec{a}} = C'_{v,\vec{a}}$.
By induction,  we know that  $C''_{u,\vec{a}} = C''_{v,\vec{a}}$ is 
an $\vec{r}$-clause.  Since $x$ is a $\vec{q}$-variable and is not
instantiated by $\vec{a}$,  it must be
the case that $x\not\in C'_{v,\vec{a}}$.
Thus $C'_{u,\vec{a}} = C'_{v,\vec{a}} \preceq C_v \setminus \{x\}
  \preceq C_u$.
  
Else if $g_w(\vec{a}) =1$, define $C'_{u,\vec{a}} = C'_{w,\vec{a}}$.
By a similar analysis as above, $C'_{u,\vec{a}} = C'_{w,\vec{a}}
\preceq C_w \setminus \{\neg x\}
  \preceq C_u$.

If $g_v(\vec{a}) = g_w(\vec{a}) = 0$, and if $x \notin
C'_{v,\vec{a}}$, define $C'_{u,\vec{a}} = C'_{v,\vec{a}}$. Otherwise, if
$\neg x \notin C'_{w,\vec{a}} $, define $C'_{u,\vec{a}} =
C'_{w,\vec{a}}$. It follows from induction that $C'_{u,\vec{a}}
\preceq C_u$. 

Else, define $C'_{u,\vec{a}}$ to be the resolvent of
$C'_{v,\vec{a}}$ and $C'_{w,\vec{a}}$ on $x$. 
By induction, we know that $C'_{v,\vec{a}} \setminus \{x\}\preceq C_1
\vee U_1$ and $C'_{w,\vec{a}}\setminus \{\neg x\} \preceq C_2 \vee
U_2$. Hence 
$C'_{u,\vec{a}} \preceq C_1 \vee C_2 \vee U =C_u$. 

We need to verify the conditions on $g_u(\vec{a})$ and $C''_{u,\vec{a}}$.
The case when $g_u(\vec{a}) = 1$ is immediate,
since $C''_{u,\vec{a}}$ copies a clause known by induction to be an
$\vec{r}$-clause. So now
consider  the case when
$g_u(\vec{a}) = 0$.
By induction, we know that both $C''_{v,\vec{a}} =
C'_{v,\vec{a}}|_{\vec{a}}$ and $C''_{w,\vec{a}} =
C'_{w,\vec{a}}|_{\vec{a}}$ are $\vec{q}$-clauses and can be derived
using $A(\vec{a}, \vec{q})$ alone via \lquprc.  

We have three cases. If $C'_{u,\vec{a}} = C'_{v,\vec{a}}$ or
$C'_{u,\vec{a}} = C'_{w,\vec{a}}$, then by induction we are done.
Otherwise, $C'_{u,\vec{a}}$ is obtained from $C'_{v,\vec{a}}$ and
$C'_{w,\vec{a}}$ via a resolution step on pivot $x$.  Since
$\vec{a}$ is an assignment to the $\vec{p}$ variables and $x \notin
\vec{p}$, $C''_{u,\vec{a}}$ can be derived from $C''_{v,\vec{a}}$ and
$C''_{w,\vec{a}}$ via the same resolution step. 

\noindent
{\bf Note:} A simple observation is that $C'_{u,\vec{a}}$ is always a
subset of $C_u$ with only one exception, which is that some special
symbol $u^*$ in $C_u$ may be converted into $u$ in
$C'_{u,\vec{a}}$. This leads us to define the relation
$\preceq$. Also, the resolution step in $\pi''(\vec{a})$ is applicable in
\lquprc because (1)~every mergable universal variable in
$C''_{v,\vec{a}}$ and $C''_{w,\vec{a}} $ was also mergable earlier in
$C_v$ and $C_w$ in $\pi$. (2)~Every common non-mergable existential
variable in $C''_{v,\vec{a}}$ and $C''_{w,\vec{a}}$ was also a
non-mergable existential variable in $C_v$ and $C_w$. (3)~Every
non-mergable universal variable in $C''_{v,\vec{a}}$ and
$C''_{w,\vec{a}}$ was also a non-mergable universal pair in $C_v$ and
$C_w$. (4)~The operations do not disturb the indices of variables,
therefore if variable $x$ satisfies the index condition in $\pi$ it
satisfies it in $\pi''(\vec{a})$ as well.

\medskip\noindent
{\bf At an internal node with $\vec{r}$-resolution:} 
  Let node $u$ in $\pi$ correspond to a resolution step with pivot $x \in \vec{r}$. This is dual to the case above.
\end{proof}


\subsection{Interpolants from \irmc proofs}
\label{subsec:irm-interpolant}
We now establish the interpolation theorem for the expansion-based 
calculi, following the same overall idea described in
Section~\ref{subsec:interpolation-setting}.

\begin{theorem} \label{thm:irm}
 \irmc  has feasible interpolation.
\end{theorem}
\begin{proof}
  This proof closely follows that of Theorem~\ref{thm:lqup}, but
  with several  changes in the proof details. We describe the changes
  here.

{\bf Construction of the circuit $C_{\pi}$:}
The circuit construction is very similar to  that for
\lquprc. Leaves and resolution nodes are treated as before.
Instantiation and merging nodes are treated as the universal reduction
nodes were; that is, the corresponding gates are  no-operation
(identity) gates. 

{\bf Construction of $\pi'$ and $\pi''$:} 
As before we construct a proof-like structure $\pi'(\vec{a})$,
which depends on the assignment $\vec{a}$ to the $\vec{p}$ variables,
the proof $\pi$ of $\mathcal{F}$, and the circuit $C_{\pi}$.  For each
node $u$ in $\pi$, with clause $C_u$, we associate a clause
$C'_{u,\vec{a}}$ in $\pi'(\vec{a})$, and  let $C''_{u,\vec{a}}$ be the
instantiation of $C'_{u,\vec{a}}$  by the assignment $\vec{a}$. We
show (by induction on the height of $u$ in $\pi$) that:  
\begin{enumerate}
	\item $C'_{u,\vec{a}} \preceq C_u$.

	\item $g_{u}(\vec{a}) = 0 \implies C''_{u,\vec{a}}$ is a $\vec{q}$-clause and can be obtained from the clauses of 
			$A(\vec{a},\vec{q})$ alone using the rules of system \irmc.

	\item $g_{u}(\vec{a}) = 1 \implies C''_{u,\vec{a}}$ is a $\vec{r}$-clause and can be obtained from the clauses of
			$B(\vec{a},\vec{r})$ alone using the rules of system \irmc.
\end{enumerate} 
Once again, as described in the proof outline, this suffices to conclude that the circuit $C_\pi$ computes an interpolant.

Recall that for annotated clauses,
the meaning of $\preceq$ is slightly different and is  given in Definition~\ref{def:preceq}.

\medskip\noindent
{\bf At a leaf level: } Let node $u$ be a leaf in $\pi$. Then
$C'_{u,\vec{a}} = C_u$; that is, copy the clause as it is. Trivially,
$C'_{u,\vec{a}} \preceq C_u$.
By construction of $C_\pi$, the
conditions concerning $g_u(\vec{a})$ and $C''_{u,\vec{a}}$ are satisfied. 

%
\medskip\noindent
{\bf At an internal node with instantiation:} 
  Let node $u$ be an internal node in $\pi$ corresponding to an instantiation step by $\tau$. And let node $v$ be its only parent.
We know $C_u=\instantiate(\tau, C_v)$. 

Suppose $l^{\sigma'} \in \instantiate(\tau, C'_{v,\vec{a}})$. Then for some $\xi'$, $l^{\xi'} \in
C'_{v,\vec{a}}$, and $l^{\sigma'} = l^{[\complete{\xi'}{\tau}]}$; hence
$\sigma'$ is a subset of $\xi'$ completed with $\tau$. By
induction we know that $C'_{v,\vec{a}}\preceq C_v$. We have an injective function $f:C'_{v,\vec{a}}\hookrightarrow C_v$ that demonstrates this. Let $f(l^{\xi'})=l^\xi$. Hence  $l^\xi\in
C_v$ for some $\xi'\preceq\xi$.
So $l^{\sigma} = l^{[\complete{\xi}{\tau}]}\in C_u$.
Since the annotations introduced by instantiation match,
$\sigma'\preceq \sigma $. We use this to define a function
$g:\instantiate(\tau, C'_{v,\vec{a}})\rightarrow C_u$ where
$g(l^{\sigma'})=l^{\sigma}$. Now we find any  $l^{\tau_1},l^{\tau_2}$
where $g(l^{\tau_1})=g(l^{\tau_2})=l^\tau$ and perform a merging step
on $l^{\tau_1}$ and $l^{\tau_2}$; note that the resulting literal
$l^{\tau'}$ will still satisfy $\tau'\preceq \tau$. Eventually we get a clause which we define as $\minst(\tau, C'_{v,\vec{a}},C_u)=C'_{u,\vec{a}}$ where this function is injective. We will use this notation to refer to this process of instantiation and then deliberate merging to get $\preceq C_u$.
 
Therefore $C'_{u,\vec{a}}\preceq C_u$. 

If the node $u$ is not pruned out in $\pi''(\vec{a})$, then
$C''_{u,\vec{a}}$  contains no satisfied $\vec{p}$ literals; hence
neither does $C'_{v,\vec{a}}$. Therefore $C''_{u,\vec{a}}$ is derived from  $C''_{v,\vec{a}}$; this is a valid step in the proof
system. 

Because we only use instantiation and merging or a dummy step, $C''_{u,\vec{a}}$
is a $\vec{q}$-clause if and only if $C''_{v,\vec{a}}$ is a
$\vec{q}$-clause. Therefore the no-operation (identity) gate $g_u$ gives a valid
result by induction. 

\medskip\noindent
{\bf At an internal node with merging:} 
  Let node $u$ be an internal node in $\pi$ corresponding to a merging step. 
Let node $v$ be its only parent.
We have $$\frac{C_v=D_v\vee b^\mu\vee b^\sigma}{C_u=D_v\vee b^\xi}$$
where $\domain(\mu)=\domain(\sigma)$ and 
$\xi$ is obtained by merging the annotations $\mu,\sigma$. That is,
$\xi = \AMerge(\mu,\sigma) = \merge{\mu}{\sigma}$. 
Note that $\mu, \sigma \preceq \AMerge(\mu,\sigma)$. 

Note that from the induction hypothesis, $C'_{v,\vec{a}}\preceq C_v$, so there is an injective function $f:C'_{v,\vec{a}}\hookrightarrow C_v$. Suppose $C'_{v,\vec{a}}$ contains two distinct literals $b^{\mu'}$
and $b^{\sigma'}$ where $f(b^{\mu'})=b^{\mu}$ and $f(b^{\sigma'})=b^{\sigma}$.
%
So $C'_{v,\vec{a}} = D'_v\vee b^{\mu'}\vee b^{\sigma'}$.
Then let 
$C'_{u,\vec{a}}=D'_v\vee b^{\xi'}$, where
$\xi' = \AMerge(\mu',\sigma')$. 
Otherwise let $C'_{u,\vec{a}}=C'_{v,\vec{a}}$.

We first observe whenever we do actual merging, if $c/u\in \xi'$ then
one of the following holds:
\begin{enumerate}
\item $c/u\in \sigma'$. Then $c/u \in \sigma$ or $*/u\in \sigma$,  and
so $c/u \in \xi$ or $*/u\in \xi$. 
\item $c/u\in \mu'$. Then $c/u \in \mu$ or $*/u\in \mu$, and so $c/u
\in \xi$ or $*/u\in \xi$. 
\item $e/u\in\mu'$, $d/u\in\sigma'$, $e\neq d$, in which case $*/u\in \xi$. 
\end{enumerate}
Since all other annotated literals are unaffected,
$C'_{u,\vec{a}}\preceq C_u$.
We never merge $\vec{p}$ literals as they have no annotations, so if
$C''_{u,\vec{a}}$ is not pruned away, then 
$C''_{u,\vec{a}}$ is derived from $C''_{v,\vec{a}}$ via merging.

In case we do not merge, there might be some $b^{\sigma'}\in
C'_{v,\vec{a}}$ with $\sigma'\preceq\sigma$, which is not removed by
merging. 
However $\sigma'\preceq\sigma\preceq \xi$, so
$C'_{u,\vec{a}}=C'_{v,\vec{a}}\preceq C_u$. As
$C''_{u,\vec{a}}=C''_{v,\vec{a}}$, this is a valid inference step (in
fact, a dummy step). 

Because we only use merging or a dummy step, $C''_{u,\vec{a}}$ is a
$\vec{q}$-clause if and only if $C''_{v,\vec{a}}$ is a
$\vec{q}$-clause, therefore the no-operation (identity) gate $g_u$ gives a valid result by induction.

\medskip\noindent
{\bf At an internal node with $\vec{p}$-resolution:} 
We do not have any annotations on $\vec{p}$-literals. So in this case
we construct $C'_u$ and $C''_u$ exactly as we would for an \lquprc proof.

\medskip\noindent
{\bf At an internal node with $\vec{q}$-resolution:} 
  When we have a resolution step between nodes $v$ and $w$ on a
  $\vec{q}$ pivot to get node $u$, we
  have $$\frac{C_v=x^{\tau\cup\xi}\lor D_v \hspace{5mm} C_w=\lnot
    x^{\tau\cup\sigma}\lor
    D_w}{C_u=\instantiate(\sigma,D_v)\cup\instantiate(\xi,D_w)}$$
  where $\domain(\tau)$, $\domain(\xi)$ and $\domain(\sigma)$ are
  mutually disjoint, and $\range(\tau) \subseteq \{0,1\}$. 

In order to do dummy instantiations we will need to define a $\{0,1\}$ version of $\xi$ and $\sigma$. So we define $\xi'=\{c/u \mid c/u \in \xi , c\in\{0,1\} \}\cup\{0/u \mid */u \in \xi \}$, $\sigma'=\{c/u \mid c/u \in \sigma , c\in\{0,1\} \}\cup\{0/u \mid */u \in \sigma \}$. This gives us the desirable property that  $\xi'\preceq\xi$, $\sigma'\preceq\sigma$. 

Now resuming the construction of $C'$, we use information from the
circuit to construct this. If $g_v(\vec{a})=1$, then we define
$C'_{u,\vec{a}}=\minst(\sigma',C'_{v,\vec{a}}, C_u )$. Otherwise, if
$g_w(\vec{a})=1$, then we define
$C'_{u,\vec{a}}=\minst(\xi',C'_{w,\vec{a}}, C_u )$. In these cases,
we know by the inductive claim that $C'_{u,\vec{a}}$ does not contain
any $\vec{q}$ literals. Therefore $C'_{u,\vec{a}}$ is the correct
instantiation (as $\xi'\preceq\xi$, $\sigma'\preceq\sigma$) of some
subset of $D_v$ or $D_w$. Hence $C'_{u,\vec{a}}\preceq C_u$. Furthermore since $g_u$ is an OR gate evaluating to 1 and since $C''_{u, \vec{a}}$, an $\vec{r}$-clause, can be obtained by an instantiation step, our inductive claim is true.

Now suppose $g_v(\vec{a})=0$ and $g_w(\vec{a})=0$. If there is no
$x^{\mu} \in C'_{v,\vec{a}}$ such that $ \mu\preceq \tau\cup\xi$, then
define $C'_{u,\vec{a}}=\minst(\sigma',C'_{v,\vec{a}}, C_u )$. Else, if
there is no $\neg x^{\mu} \in C'_{w,\vec{a}}$ such that $ \mu\preceq
\tau\cup\sigma$, then define
$C'_{u,\vec{a}}=\minst(\xi',C'_{w,\vec{a}}, C_u )$. In these cases we
know that $C'_{u, \vec{a}}$ is the correct instantiation (as
$\xi'\preceq\xi$, $\sigma'\preceq\sigma$) of some subset of $D_v$ or
$D_w$; hence $C'_{u,\vec{a}}\preceq C_u$. Furthermore, since $g_u$ is
an OR gate evaluating to 0, and since $C''_{u,\vec{a}}$, a $\vec{q}$-clause, can be obtained by an instantiation step, our inductive claim is true.

The final case is when $g_v(\vec{a})=g_w(\vec{a})=0$ and $x^{\tau \cup
  \xi_1} \in C'_{v,\vec{a}}$ for some $ \xi_1\preceq \xi $ and $\neg
x^{\tau \cup\sigma_1} \in C'_{w,\vec{a}}$ for some $ \sigma_1\preceq
\sigma $. Here, because $\domain(\tau)$, $\domain(\xi)$ and
$\domain(\sigma)$ are mutually disjoint, $\domain(\tau)$,
$\domain(\xi_1)$ and $\domain(\sigma_1)$ are also mutually
disjoint. Thus we can do the resolution step


$$\frac{C'_{v,\vec{a}}=x^{\tau\cup\xi_1}\lor D'_v \hspace{5mm} C'_{w,\vec{a}}=\lnot x^{\tau\cup\sigma_1}\lor D'_w}{\instantiate(\sigma_1 ,D'_v)\cup\instantiate(\xi_1,D'_w)}.$$


Since $\minst(\sigma_1 ,D'_v, C_u)\preceq \instantiate(\sigma,D_v)$
and $\minst(\xi_1 ,D'_w, C_u)\preceq \instantiate(\xi,D_w)$, we can
follow up $\instantiate(\sigma_1 ,D'_v)\cup\instantiate(\xi_1,D'_w)$
with sufficient merging steps to get a clause $C' \preceq C_u$; we
define this clause to be the clause $C'_{u,\vec{a}}$. 
By the inductive claim, both
$C''_{v,\vec{a}}$ and $C''_{w,\vec{a}}$ are $\vec{q}$-clauses; hence
$C''_{u,\vec{a}}$ is also a $\vec{q}$-clause and is obtained via a
valid resolution step.

\medskip\noindent
{\bf At an internal node with $\vec{r}$-resolution:} 
  When we have a resolution step between nodes $u$ and $v$ on an $\vec{r}$-literal, this is the dual of the previous case.
\end{proof}

\subsection{Monotone Interpolation}

To transfer known circuit lower bounds into size of proof bounds, we need a monotone version of the previous interpolation theorems, which we prove next.

\begin{theorem} \label{theorem2}
   \lquprc and \irmc have monotone feasible interpolation.
\end{theorem} 
\begin{proof}
In previous subsections, we have shown that the circuit
$C_{\pi}(\vec{p})$ is a correct interpolant for the QBF 
$\mathcal{F}$. That is, if $C_{\pi}(\vec{p}) = 0$ then $\mathcal{Q}
\vec{q}. A(\vec{a}, \vec{q})$ is false, and if $C_{\pi}(\vec{p}) = 1$
then $\mathcal{Q} \vec{r}. B(\vec{a}, \vec{r})$ is false.

However, if $\vec{p}$ occurs only positively in $A(\vec{p}, \vec{q})$ then we construct a monotone circuit $C^\mon_{\pi}(\vec{p})$ such that, on every $0, 1$ assignment $\vec{a}$ to $\vec{p}$ we have
\begin{align*} 
C^\mon_{\pi}(\vec{a}) = 0 &\implies \mathcal{Q} \vec{q}. A(\vec{a}, \vec{q}) \text{ is false, and }  \\
C^\mon_{\pi}(\vec{a}) = 1 &\implies \mathcal{Q} \vec{r}. B(\vec{a}, \vec{r}) \text{ is false. }
\end{align*}

We obtain $C^\mon_{\pi}(\vec{p})$ from $C_{\pi}(\vec{p})$ by replacing all selector gates $g_u = \sel(x, g_v, g_w)$ by the following monotone ternary connective: $g_u = (x \vee g_v) \wedge g_w$ where nodes $v$ and $w$ are the parents of $u$ in $\pi$.
We also change the proof-like structure $\pi'(\vec{a})$; the
construction is the same as before except that at $\vec{p}$-resolution nodes, 
the rule for fixing $C'_{u,\vec{a}}$ is also changed to reflect the
monotone function used instead.

More precisely, the functions $\sel(x, g_v, g_w)$ and
$g_u = (x \vee g_v) \wedge g_w$ differ only when $x = 0$,
$g_v(\vec{a}) = 1$, and 
$g_w(\vec{a}) = 0$.  We set $C'_{u,\vec{a}}$ to 
$C'_{w,\vec{a}}\setminus \{\neg x\}$ if $x=1$ or if $x=0$, $g_v(\vec{a})=1$ and
$g_w(\vec{a})=0$, and to  $C'_{v,\vec{a}}\setminus \{x\}$ otherwise. 

It suffices to verify the inductive statements  in the case when
$x=0$, $g_v(\vec{a}) = 1$, and $g_w(\vec{a}) = 0$. We have to show
that $C'_{u,\vec{a}} \preceq C_u$; this holds by induction.
We also have to show that 
$C''_{u,\vec{a}}$ is a $\vec{q}$-clause, and can be derived using
$A(\vec{a}, \vec{q})$ clauses alone via the appropriate proof system.
By induction, since $g_w(\vec{a}) = 0$, we conclude that 
$C''_{w,\vec{a}}$ cannot contain $\neg x$: it can be derived
from the clauses of $A(\vec{p},\vec{q})$ alone,
by the positivity constraint, these clauses do not contain $\neg x$, and the derivation cannot introduce literals.  
Hence $C''_{w,\vec{a}}=
C'_{w,\vec{a}} \setminus \{\neg x\}|_{\vec{a}}$, which is
 $C''_{u,\vec{a}}$.
\end{proof}

\section{New Exponential Lower Bounds for \irmc and \lquprc}
\label{sec:lower-bounds}

We now apply our interpolation theorems to obtain new exponential lower bounds for a new class of QBFs. The lower bound will be directly transferred from the following monotone circuit lower bound for the problem $\clique(n,k)$, asking whether a given graph with $n$ nodes has a clique of size $k$. 
\begin{theorem}[Alon \& Boppana \cite{AB87}] \label{thm:raz}
  All monotone circuits that compute $\clique(n,n/2)$ are of exponential size. 
\end{theorem}

We now build the QBF. Fix an integer $n$ (indicating the number of
vertices of the graph) and let $\vec{p}$ be the set of variables
$\{p_{uv} \mid  1\leq u<v\leq n\}$. An assignment to $\vec{p}$ picks
a set of edges, and thus an $n$-vertex graph.
Let $\vec{q}$ be the set of variables $\{q_{iu} \mid i \in [\frac{n}{2}], u \in [n] \}$.
We use the following clauses.
\[
\begin{array}{c@{\ }c@{\ }ll@{\ }c@{\ }ll}
C_i&=&q_{i1}\vee \dots \vee q_{in} & \quad \text{for } i \in [\frac{n}{2}]\\
D_{i,j,u}&=&\neg q_{iu}\vee \neg q_{ju} & \quad \text{for } i, j \in [\frac{n}{2}], i< j \text{ and }  u\in [n]\\
E_{i,u,v}&=&\neg q_{iu}\vee \neg q_{iv} & \quad \text{for } i\in [\frac{n}{2}] \text{ and } u, v\in [n], u<v\\
F_{i,j,u,v}&=&\neg q_{iu}\vee \neg q_{jv}\vee p_{uv} & \quad \text{for
} i, j \in [\frac{n}{2}], i\neq j$ \text{ and } $u, v\in [n], u < v.
\end{array}
\]
(For notational convenience, we interpret the assignment $p_{uv}=0$ to
mean that the edge $uv$ is present in the graph.)

We can now express $\clique(n,n/2)$ as a polynomial-size QBF $\exists \vec{q}. A_n(\vec{p},\vec{q})$, where 
$$
  A_n(\vec{p},\vec{q})=\bigwedge_{i \in [\frac{n}{2}]}C_i\wedge \bigwedge_{i<j,u\in [n]}D_{i,j,u}\wedge \bigwedge_{i\in [\frac{n}{2}],u<v}E_{i,u,v}\wedge \bigwedge_{i<j,u\neq v}F_{i,j,u,v}.
$$
Here the edge variables $\vec{p}$ appear positively in $A_n(\vec{p},\vec{q})$.

Likewise no-$\clique(n,n/2)$ can be written as a polynomial-size QBF $\forall \vec{r_1}\exists \vec{r_2}. B_n(\vec{p},\vec{r_1}, \vec{r_2})$. To construct this we use a polynomial-size circuit  that checks whether the nodes specified by $\vec{r_1}$ fail to form a clique in the graph given by $\vec{p}$.  We then use existential variables $\vec{r_2}$ for the gates of the circuit and can then form a CNF $B_n(\vec{p},\vec{r_1}, \vec{r_2})$ that represents the circuit computation.

Now we can form a sequence of false QBFs, stating that the graph encoded in $\vec{p}$ both has a clique of size $n/2$ (as witnessed by $\vec{q}$) and likewise does not have such a clique as expressed in the $B$ part:
$$
  \Phi_n=\exists \vec{p}\exists \vec{q}\forall \vec{r_1}\exists \vec{r_2}. A_n(\vec{p},\vec{q})\wedge B_n(\vec{p},\vec{r_1}, \vec{r_2}).
$$
This formula has the unique interpolant $\clique(n,n/2)(\vec{p})$. But
since all monotone circuits for this are of exponential size by
Theorem~\ref{thm:raz}, and since monotone circuits of size polynomial in \irmc and \lquprc proofs 
can be extracted by Theorem~\ref{theorem2}, all such proofs must be of exponential size, yielding:

\begin{theorem} \label{thm:lower-bound-clique}
   The QBFs $\Phi_n(\vec{p},\vec{q},\vec{r})$ require exponential-size proofs in \irmc and \lquprc.
\end{theorem}

Note: A slightly different, and arguably more transparent, way of
encoding no-$\clique(n,n/2)$ is described in \cite{BCMS-FST16}. 

\section{Feasible Interpolation vs.\ Strategy Extraction}
\label{sec:strat-extraction}

Recall the two player game semantics of a QBF explained in Section~\ref{sec:prelim}. Every false QBF has a
winning strategy for the universal player, where the strategy value for each
variable depends only on the values of the variables played before. We
now explain 
the relation between strategy extraction --- one of the main paradigms
for QBF systems --- and feasible interpolation.
In Section~\ref{sec:interpolation} we studied QBFs of the form
$\mathcal{F}= \exists \vec{p} \mathcal{Q} \vec{q} \mathcal{Q} \vec{r}. \left[A(\vec{p}, \vec{q}) \wedge B(\vec{p}, \vec{r})\right].$ If we  add a common universal variable $b$ we can change it to an equivalent QBF
$$
  \mathcal{F}^b= \exists \vec{p}\, \forall b\,  \mathcal{Q} \vec{q}\, \mathcal{Q} \vec{r}. \left[ (A(\vec{p}, \vec{q})\vee  b) \wedge  (B(\vec{p}, \vec{r})\vee \neg b) \right].
  $$
  This can be expressed with a CNF matrix by inserting the literal $b$
  into each clause of $A(\vec{p}, \vec{q})$ and the literal $\neg b$
  into each clause of $B(\vec{p}, \vec{r})$. Let  $\mathcal{F}^b$ also
  denote this equivalent QBF. 
  
If $\mathcal{F}$ is false, then also $\mathcal{F}^b$ is false and thus the universal player has a winning strategy, including a strategy for $b=\sigma(\vec{p})$ for the common universal variable $b$.

\begin{remark}
Every winning strategy $\sigma(\vec{p})$ for $b$ is an interpolant for $\mathcal{F}$, i.e., for every $0, 1$ assignment $\vec{a}$ of $\vec{p}$ we have
\begin{align*} 
\sigma(\vec{a}) = 0 &\implies \mathcal{Q} \vec{q}. A(\vec{a}, \vec{q}) \text{ is false, and }\\  
\sigma(\vec{a}) = 1 &\implies \mathcal{Q} \vec{r}. B(\vec{a}, \vec{r}) \text{ is false. }
\end{align*}
\end{remark}

\begin{proof}
  Suppose not. Then there are two possibilities. 
  \begin{itemize} 
\item There is some $\vec{a}$ where $\sigma(\vec{a})=0$ and $\mathcal{Q} \vec{q}. A(\vec{a}, \vec{q})$ is true. Then setting $b=0$ would satisfy $\mathcal{Q} \vec{r}. B(\vec{p}, \vec{r})\vee \neg b$. But $\mathcal{Q} \vec{q}.A(\vec{a}, \vec{q}) \vee b$ is also satisfied. Hence this cannot be part of the winning strategy for the universal player.
\item There is some $\vec{a}$ where $\sigma(\vec{a})=1$ and
  $\mathcal{Q} \vec{r}. B(\vec{a}, \vec{r})$ is true. This is the dual
  of the above.  
  \qedhere
\end{itemize}
\end{proof}

This observation means that every interpolation problem can be reformulated as a strategy extraction problem. We will now show that from proofs of these reformulated interpolation problems we can extract a (monotone) Boolean circuit for the winning strategy on the new universal variable $b$.

Strategy extraction was recently shown to be a powerful lower bound technique for QBF resolution systems. 
In strategy extraction, from a refutation of a false QBF, winning strategies for the universal player for all universal variables can be efficiently extracted. Devising QBFs that require computationally hard strategies then leads to lower bounds for QBF proof systems.  
This technique applies both to \qrc \cite{BCJ15}, where \AC{0} lower bounds for e.g.\ parity are used, as well as to much stronger QBF Frege systems where the full spectrum of current (and conjectured) lower circuit bounds is employed \cite{BBC16}. In fact, Beyersdorff and Pich \cite{BP16} show that lower bounds for QBF Frege systems can only come either (a) from  circuit lower bounds via the strategy extraction technique or (b) from lower bounds for classical proposition Frege. This picture is reconfirmed here as well: QBF resolution lower bounds via feasible interpolation fall under paradigm (a) as they are in fact lower bounds via strategy extraction.

To show this we now prove how to extract strategies for interpolation problems, first for \lquprc and then for \irmc.

\begin{theorem} \label{thm:b-lqu} \hfill
\begin{enumerate}
\item From each \lquprc refutation $\pi$ of $\mathcal{F}^b$ we can extract in polynomial time a boolean circuit for $\sigma(\vec{p})$, i.e., the part of the winning strategy for variable $b$.
\item If in the same setting as above for $\mathcal{F}^b$, the variables $\vec{p}$ appear only positively in $A(\vec{p}, \vec{q})$, then we can extract a monotone boolean circuit for $\sigma(\vec{p})$ from a \lquprc refutation $\pi$ of $\mathcal{F}^b$ in polynomial time (in the size of $\pi$).
\end{enumerate}
\end{theorem}

\begin{proof}
As we can compute the (monotone) interpolant when $b$ is absent, we use the same proof with a few modifications for the new formula.

We first change the definition of $\vec{q}$ and $\vec{r}$-clauses to allow for $b$ and $\neg b$ literals.
\begin{definition}
We call any clause in the proof a $\vec{q}$-clause (resp.\ $\vec{r}$-clause) if it contains only variables $\vec{p}, \vec{q}$ or literal $b$ (resp. $\vec{p}, \vec{r}$ or literal $\neg b$). We retain the inheritance property for clauses only containing $\vec{p}$ variables.
\end{definition} 

{\bf Construction of the circuit $C_{\pi}$:}
When constructing the circuit, we now also need to consider a  resolution
step on the common universal variable $b$:
$$
\frac{C_v = \overbrace{C_1 \vee U_1 \vee b}^\text{node $v$} \hspace{7mm} \overbrace{C_2 \vee U_2 \vee \neg b}^\text{node $w$} = C_w}{C_u = \underbrace{C_1 \vee C_2 \vee U}_\text{node $u$}}. 
$$
Here we can arbitrarily pick one of $v$ or $w$. For example here we
pick $v$ and let $g_u$ be wired to $g_v$ with the no-operation
(identity) gate, disregarding the input from $g_w$.

{\bf Construction of $\pi'$ and $\pi''$:} 
We slightly modify the invariants to include the new definitions.  Additionally we make a small change to the first invariant.

\begin{enumerate}
	\item $C'_{u,\vec{a}}\backslash\{b, \neg b\} \preceq C_u$.

	\item $g_{u}(\vec{a}) = 0 \implies C''_{u,\vec{a}}$ is a $\vec{q}$-clause and can be obtained from the clauses of 
			$A(\vec{a},\vec{q})$ alone using the rules of  \lquprc.

	\item $g_{u}(\vec{a}) = 1 \implies C''_{u,\vec{a}}$ is a $\vec{r}$-clause and can be obtained from the clauses of
			$B(\vec{a},\vec{r})$ alone using the rules of  \lquprc.
\end{enumerate}

Notice also that $b^* \notin C''_{u,\vec{a}}$ as $b^*$ can only arise
from a long distance resolution step on a $\vec{p}$ variable but these
are instantiated and so never occur as pivots in the proof $\pi''$ assuming the induction hypothesis.

We observe that the base cases work for the construction of $\pi'$ and $\pi''$. The only new part of the inductive step is when we have
$$\frac{C_v = \overbrace{C_1 \vee U_1 \vee b}^\text{node $v$} \hspace{7mm} \overbrace{C_2 \vee U_2 \vee \neg b}^\text{node $w$} = C_w}{C_u = \underbrace{C_1 \vee C_2 \vee U}_\text{node $u$}}. 
$$
To find $C'_{u,\vec{a}}$ we look at our choice of wiring in the circuit construction. If $g_u$ is wired to $g_v$ ($g_u=g_v$) then we take $C'_{u,\vec{a}}$ to equal $C'_{v,\vec{a}}$. Since $C'_{v,\vec{a}}\backslash\{b, \neg b\}\preceq C_v\backslash\{b, \neg b\}\preceq C_u$ we get $C'_{u,\vec{a}}\backslash\{b, \neg b\}\preceq C_u$.
Since our choice of the clause is determined by our choice of wiring, then we retain our invariants in that way.

Notice that we never resolve a $\vec{q}$-clause with a $\vec{r}$ clause in $\pi''$ so $b, \neg b$ will always be retained in their respective type of clauses.

From the above, we have the following conclusion. Let $r$ be the root of $\pi$. Then on any assignment $\vec{a}$ to the $\vec{p}$ variables we have:

\begin{description}
	\item[(1)] $C'_{r,\vec{a}} \backslash \{b, \neg b\} \preceq
          C_r = \Box$. Therefore, $C''_{r,\vec{a}}\setminus \{b, \neg
          b\} = \Box$. But $C''_{r,\vec{a}}$ can contain at most one
          of these literals, which can be universally reduced to
          complete a refutation. 

	\item[(2)] $g_{r}(\vec{a}) = 0 \implies  C''_{r,\vec{a}}$ is a $\vec{q}$-clause and can be obtained from the clauses 			of $A(\vec{a},\vec{q})$ alone using the rules of system \lquprc. Hence by soundness of \lquprc, 
		$\mathcal{Q} \vec{q}. A(\vec{a}, \vec{q})$ is false.

	\item[(3)] $g_{r}(\vec{a}) = 1 \implies  C''_{r,\vec{a}}$ is an $\vec{r}$-clause and can be obtained from the clauses 			of $B(\vec{a},\vec{r})$ alone using the rules of system \lquprc. Hence by soundness of \lquprc,
		$\mathcal{Q} \vec{r}. B(\vec{a}, \vec{r})$ is false.
\end{description}
\noindent
Thus $g_r$, the output gate of the circuit, computes $\sigma(\vec{p})$. 
\end{proof}

An analogous result to Theorem~\ref{thm:b-lqu} also holds for \irmc.

\begin{theorem} \label{thm:b-irmc} \hfill
\begin{enumerate}
\item From each \irmc refutation $\pi$ of $\mathcal{F}^b$ we can extract in polynomial time a boolean circuit for $\sigma(\vec{p})$, i.e., the part of the winning strategy for variable $b$.
\item If in the same setting as above for $\mathcal{F}^b$, the variables $\vec{p}$ appear only positively in $A(\vec{p}, \vec{q})$, then we can extract a monotone boolean circuit for $\sigma(\vec{p})$ from a \irmc refutation $\pi$ of $\mathcal{F}^b$ in polynomial time (in the size of $\pi$).
\end{enumerate}
\end{theorem}
\begin{proof}
We can use exactly the same constructions as in Theorem~\ref{thm:irm}. The $b$ literals do not affect the argument.
%
%
%
%
%
%
\end{proof}

As a corollary, the versions $\Phi^b_n(\vec{p},\vec{q},\vec{r})$ of the formulas from Section~\ref{sec:lower-bounds} also require exponential-size proofs in \irmc and \lquprc.

\section*{Acknowledgements}
We thank Pavel \Pudlak and Mikol\'a\v{s} Janota for helpful
discussions on the relation between feasible interpolation and
strategy extraction during the Dagstuhl Seminar `Optimal
Agorithms and Proofs' (14421).


\newcommand{\etalchar}[1]{$^{#1}$}

\end{document}